\documentclass[a4paper,10pt]{article}

\usepackage{a4wide}
\usepackage{amsmath,amsfonts,amssymb,amsthm,dsfont}

\usepackage[active]{srcltx}

\usepackage{mathrsfs}
\usepackage{amscd}

\newcounter{theorem}
\renewcommand{\thetheorem}{\arabic{section}.\arabic{theorem}}

\newenvironment{thm}[1]{\par
\begin{sloppypar}\refstepcounter{theorem}%
\noindent{\bf #1 \thetheorem.}\it{}}{\end{sloppypar}}

\newenvironment{theorem}{\begin{thm}{Theorem}}{\end{thm}}
\newenvironment{proposition}{\begin{thm}{Proposition}}{\end{thm}}
\newenvironment{corollary}{\begin{thm}{Corollary}}{\end{thm}}

\newenvironment{defi}[1]{\par
\begin{sloppypar}\refstepcounter{theorem}%
\noindent{\bf #1 \thetheorem.}\rm{}}{\end{sloppypar}}
\newenvironment{definition}{\begin{defi}{Definition}}{\end{defi}}

\newenvironment{remark}{\begin{defi}{Remark}}{\end{defi}}

\def\R{{\rm I\kern-.2em R}}

\def\X{\mathcal X}
\def\V{\mathcal{V}}

\def\s0{\sigma_0}

\def\bb1{{\rm{1}\hspace{-3pt}\mathbf{l}}}
\def\Ie0{[-\epsilon_0,\epsilon_0]}

\def\beq{\begin{equation}}
\def\eeq{\end{equation}}

\title{A Schatten-von Neumann class criterion for the magnetic Weyl calculus}
\author{Nassim Athmouni\footnote{Facult\'{e} des Sciences de
Gafsa, Gafsa, Tunisie}, Radu Purice\footnote{Institute
of Mathematics Simion Stoilow of the Romanian Academy,Research unit nr. 1; P.O.  Box
1-764, Bucharest, RO-014700, Romania and
Laboratoire Europ\'{e}en Associ\'{e} CNRS {\it Math\'{e}matique et Mod\'{e}lisation}.}}

\begin{document}

\maketitle

\begin{abstract}
We prove a criterion for a
{\it 'magnetic' Weyl operator} (see \cite{MP1,IMP1}) to be trace-class by
extending a method developed by H. Cordes \cite{Cordes}, T. Kato \cite{Kato}
and G. Arsu \cite{Arsu-JOT}.  Using the Calderon-Vaillancourt type Theorem for magnetic Weyl operators and an interpolation argument we also give a  criterion for a
{\it 'magnetic' Weyl operator} (see \cite{MP1,IMP1}) to be in a Schatten-von Neumann class with $1<p<\infty$.
\end{abstract}



\section{Introduction}

The {\it 'magnetic' Weyl quantization} \cite{MP1} is proven in \cite{MP2} to be
a {\it strict deformation quantization} in the sense of Rieffel
\cite{Rif1,Rif2,Landsman} and its associated {\it 'magnetic' Weyl calculus} is
developed in \cite{MPR1,IMP1,IMP2} where a {\it magnetic version} of the
Calderon-Vaillancourt Theorem is proven. In this paper we prove a criterion for a
{\it 'magnetic' Weyl operator} to be trace-class by
extending a method developed by H. Cordes \cite{Cordes}, T. Kato \cite{Kato}
and G. Arsu \cite{Arsu-JOT}. Using the Calderon-Vaillancourt type Theorem for magnetic Weyl operators and an interpolation argument we also give a  criterion for a
{\it 'magnetic' Weyl operator} (see \cite{MP1,IMP1}) to be in a Schatten-von Neumann class with $1<p<\infty$. Our main results are formulated in Theorem \ref{main-Th} and Corollary \ref{main-C}.

Let us fix some general notations. Recall that for any $a\geq0$ we denote by
$[a]\in\mathbb{N}$ its integer part (i.e. the largest natural number less then
or equal to $a$). For any finite dimensional real vector
space $\mathcal{V}$, we shall denote by
$BC(\V)$ the space of bounded continuous complex functions on $\mathcal{V}$ with the
$\|\cdot\|_\infty$ norm, by $C^\infty(\mathcal{V})$ the space of smooth
functions on $\mathcal{V}$, by $C^\infty_{\text{\sf
pol}}(\mathcal{V})$ its subspace of smooth functions
that are polynomially bounded together with all their derivatives and by
$BC^\infty(\V)$ the subspace of smooth functions
that are bounded together with all their derivatives; we consider all these
spaces endowed with their usual locally
convex topologies (see \cite{LSchwartz}). We use the standard multi-index notation 
\cite{Hormander2}. 
We shall consider 
the space of Schwartz test functions $\mathscr{S}(\mathcal{V})$ endowed
with its Fr\'{e}chet topology and its dual $\mathscr{S}^\prime(\mathcal{V})$
and denote by $\langle\cdot,\cdot\rangle_{\mathcal{V}}$ the associated duality
map.
We denote by $\mathcal{T}^{\mathcal{V}}_v$ the translation with $-v\in\mathcal{V}$ (acting on the
space of tempered distributions). We shall also consider the usual Sobolev spaces 
\beq
W^{m,p}(\mathcal{V}):=\left\{f\in L^p(\mathcal{V})\,\mid\,\partial^\alpha f\in L^p(\mathcal{V}),\,\forall |\alpha|\leq m\right\}
\eeq
(with $m\in\mathbb{N}$ and $1\leq p\leq\infty$) with the associated Banach space structure and by interpolation and duality also the spaces $W^{s,p}(\mathcal{V})$ for any $s\in\mathbb{R}$ \cite{B-L,T}. We shall denote by $\mathcal{H}^s:=W^{s,2}$; these are Hilbert spaces for any $s\in\mathbb{R}$.
For any vector $v\in\mathcal{V}$ we denote by $<v>:=\sqrt{1+|v|^2}$. We denote
the convolution operation by
\beq
(f*g)(v)\ :=\
\int_{\mathcal{V}}f(v-u)g(u)\,du,\qquad\forall(f,g)\in\mathscr{S}(\mathcal{V}
)\times\mathscr{S}(\mathcal{V})
\eeq
and also its possible extensions to larger spaces of distributions on
$\mathcal{V}$. For two linear topological spaces $\mathcal{L}_1$ and
$\mathcal{L}_2$ we shall denote by $\mathbb{B}(\mathcal{L}_1,\mathcal{L}_2)$
the linear space of continuous linear operators from $\mathcal{L}_1$ to
$\mathcal{L}_2$, endowed with the bounded convergence topology \cite{B}.

We shall work on the configuration space $\X:=\mathbb{R}^d$ and consider its
dual $\X^*$ with the duality map denoted by
${<\cdot,\cdot>:\X^*\times\X\rightarrow\mathbb{R}}$. Let us also consider the
\textit{phase space} $\Xi:=\X\times\X^*$ with the canonical symplectic map
$\sigma(X,Y):=<\xi,y>-<\eta,x>$ for $X:=(x,\xi)$ and $Y:=(y,\eta)$ two
arbitrary points of $\Xi$. 
We shall use some classes of H\"{o}rmander type symbols on $\Xi$. For $m\in\mathbb{R}$
and $\rho\in[0,1]$ let us define:
\beq\label{D-snorm-S}
\nu^m_{N,M}(F):=\underset{(x,\xi)\in\Xi}{\sup}<\xi>^{-m}
\underset{|\alpha|=N}{\sum}
\,\underset{|\beta|=M}{\sum}\left|\big(\partial_x^\alpha\partial_\xi^\beta
F\big)(x,\xi)\right|,\quad\forall(N,M)\in\mathbb{N}\times\mathbb{N},\ \forall
F\in C^\infty(\Xi),
\eeq
\beq
S^m_\rho(\Xi):=\left\{F\in
C^\infty(\Xi)\,\left|\,\nu^{m-M\rho}_{N,M}(F)<\infty,
\ \forall(N,M)\in\mathbb{N}\times\mathbb{N}\right\}\right..
\eeq

The\textit{Weyl quantization} (see \cite{Folland, Hormander1, Hormander2})
defines a linear and topological isomorphism
\beq\label{W-calc}
\mathfrak{Op}:\mathscr{S}^\prime(\Xi)\rightarrow\mathbb{B}\big(\mathscr{S}
(\X);\mathscr{S}^\prime(\X)\big)
\eeq
for the strong topologies. Explicitly, for $F\in\mathscr{S}(\X)$ we have the
formula
\beq
\mathfrak{Op}(F)=(2\pi)^{-d/2}\int_{\Xi}\left((2\pi)^{-d/2}\int_{\Xi}e^{
i\sigma(X,Y)}F(Y)dY\right)W(X)dX\equiv(2\pi)^{-d/2}\int_{\Xi}\mathcal{F}
_{\Xi}[F](X)W(X)dX
\eeq
\beq
\big(W((x,\xi))\phi\big)(z):=e^{(i/2)<\xi,x>}e^{-i<\xi,z>}\phi(z+x),
\quad\forall\phi\in\mathscr{S}(\X).
\eeq

We shall
usually work in the Hilbert space $\mathcal{H}:=L^2(\X)$ (defined with respect to the
Lebesgue measure). In general for a complex Hilbert space $\mathcal{K}$ we shall
denote
by $(\cdot,\cdot)_{\mathcal{K}}$ its scalar product (supposed to be
anti-linear in the first variable). For any
Hilbert space $\mathcal{K}$ we denote by
$\mathbb{B}(\mathcal{K})$ the $C^*$-algebra of bounded operators on
$\mathcal{K}$ and by $\mathbb{B}_\infty(\mathcal{K})$ its ideal of
compact operators.

\begin{definition}
Given a Hilbert space $\mathcal{K}$, for any $p\in[1,\infty)$ we consider the
linear subspace of compact operators $A\in\mathbb{B}_\infty(\mathcal{K})$ with
the property that
\beq
\exists\,\underset{N\nearrow\infty}{\lim}\,\underset{n\leq
N}{\sum}\mu_n(A)^p\,<\,\infty,
\eeq
where $\{\mu_n(A)\}_{n\in\mathbb{N}}$ are the singular values of the operator 
$A\in\mathbb{B}_\infty(\mathcal{K})$ \cite{Birman}. This subspace, denoted by
$\mathbb{B}_p(\mathcal{K})$ and called the Schatten-von Neumann class of
order $p$, is a Banach space for the norm
\beq
\|A\|_{\mathbb{B}_p(\mathcal{K})}\ :=\
\underset{N\nearrow\infty}{\lim}\,\left(\underset{n\leq
N}{\sum}\mu_n(A)^p\right)^{1/p}.
\eeq
We recall that $\mathbb{B}_1(\mathcal{K})$ is the space of
trace-class operators and $\mathbb{B}_2(\mathcal{K})$ the space of
Hilbert-Schmidt operators that is a Hilbert space for the scalar product $(
A,B)_{\mathbb{B}_2(\mathcal{K})}:=\text{\sf Tr}(A^*B)$.
\end{definition}

\subsection{The magnetic Weyl calculus.}

The magnetic fields are \textit{closed 2-forms} on $\X$ that we shall suppose to
have components
of class
$BC^\infty(\X)$. To any such magnetic field $B$ one can associate in a highly
non-unique way a vector potential $A$, i.e. a 1-form such that $B=dA$;
different choices for the vector potential are related by a change of gauge
(i.e. $dA=B=dA^\prime$ if and only if $\exists \varphi,\ A^\prime=A+d\varphi$).
We shall always suppose the vector potential to have components of class
$C^\infty_{\text{\sf pol}}(\X)$ because such a choice always exists for
magnetic fields of class $BC^\infty(\X)$.
We use two important \textit{'phase factors'} defined in terms of these
exterior forms:
\beq
\Lambda^A(x,z)\ :=\ \exp\left\{-i\int_{[x,z]}A\right\}
\eeq
\beq\label{D-Omega}
\Omega^B(x,y,z)\ :=\ \exp\left\{-i\int_{<x,y,z>}B\right\}
\eeq
where $[x,z]$ is the oriented line segment from $x\in\X$ to $z\in\X$ and
$<x,y,z>$ is the oriented triangle of vertices $\{x,y,z\}\subset\X$. From
Stokes' Theorem we deduce that
$\Omega^B(x,y,z)\ =\ \Lambda^A(x,y)\Lambda^A(y,z)\Lambda^A(z,x)$.

Let us recall from \cite{MP1} the {\it magnetic Weyl system} defined as the
family of unitary operators in $L^2(\X)$:
\beq\label{magn-W-syst}
\left\{W^A(X)\right\}_{X\in\Xi},\qquad\big(W^A((x,\xi))u\big)(z):=\Lambda^A(z,
z+x)\big(W((x,\xi))u\big)(z),\ \forall u\in\mathcal{H}.
\eeq
As explained in \cite{MP1} they are defined as \textit{unitary groups associated
to the
canonical observables in the minimal coupling formalism} for the
vector potential $A$. With the help of this magnetic Weyl system one can define
a {\it magnetic Weyl calculus} (i.e. a magnetic quantization) as in
\cite{MP1,IMP1}
\beq\label{magn-W-op}
\mathfrak{Op}^A(F)=(2\pi)^{-d/2}\int_{\Xi}\mathcal{F}
_{\Xi}[F](X)W^A(X)dX.
\eeq
Let us recall from \cite{MP1} that gauge equivalent vector potentials define unitary equivalent magnetic quantizations.

Let us make the connection with the `{\it twisted integral kernels}' formalism
in \cite{Ne05}. 
For any integral kernel ${K\in\mathscr{S}^\prime(\X\times\X)}$
let us denote by
$\mathcal{I}\text{\sf nt}K$ the corresponding linear operator on
$\mathscr{S}(\X)$; i.e. $\big(v,(\mathcal{I}\text{\sf
nt}K)u\big)_{L^2(\X)}=\langle K,\overline{v}\otimes u\rangle_{\X\times\X}$ for any
$(u,v)\in\big[\mathscr{S}(\X)\big]^2$. 
To any ${K\in\mathscr{S}^\prime(\X\times\X)}$ one can associate its 'magnetic' twisted integral kernel
\beq\label{twist-int-kernel}
K^A(x,y)\ :=\ {\Lambda}^A(x,y)K(x,y).
\eeq
Let us recall the linear bijection
$\mathfrak{W}:\mathscr{S}^\prime(\Xi)\rightarrow\mathscr{S}^\prime(\X\times\X)$
associated to the usual Weyl calculus \eqref{W-calc} by the equality
$\mathfrak{Op}(F)=\mathcal{I}\text{\sf
nt}(\mathfrak{W}F)$: 
\beq\label{iulie3}
\big(\mathfrak{W}F\big)(x,y)\ :=\
(2\pi)^{-d}\int_{\X^*}e^{i<\xi,x-y>}F\big(\frac{x+y}{2},\xi\big)d\xi.
\eeq
Then we have the equality
\beq\label{magn-quant}
\mathfrak{Op}^A(F)\ =\ \mathcal{I}\text{\sf
nt}({\Lambda}^A\mathfrak{W}F).
\eeq

This functional calculus induces a {\it magnetic Moyal product}
$\sharp^B:\mathscr{S}(\Xi)\times\mathscr{S}(\Xi)\rightarrow\mathscr{S}(\Xi)$
such that $\mathfrak{Op}^A(f\sharp^Bg)=\mathfrak{Op}^A(f)\mathfrak{Op}^A(g)$.
Explicitly we have
\beq\label{magn-Moyal-prod}
\big(f\sharp^Bg\big)\ =\
\pi^{-2d}\int_\Xi\int_\Xi\,e^{-2i\sigma(Y,Z)}\Omega^B(x-y-z,x+y-z,
x-y+z)f(X-Y)g(X-Z)\,dY\,dZ
\eeq
as oscillating integrals (see \cite{Hormander2}). We shall use the notation
\beq
\omega^B(x,y,z):=\Omega^B(x-y-z,x+y-z,x-y+z).
\eeq 
Notice that $\omega^B(x,0,z)=\omega^B(x,y,0)=1$. In \cite{IMP1} one gives the
extension of this magnetic Moyal product to the usual H\"{o}rmander type
symbols and in \cite{IMP1,IMP2} it is proven that this calculus has similar
properties with the usual Moyal product. If
$F\in\mathscr{S}^\prime(\Xi)$ is invertible for this {\it magnetic Moyal
product} we shall denote by $F^-_B$ its inverse.

In \cite{IMP1} it is proven that for any symbol $F\in S^0_0(\Xi)$ the operator
norm of $\mathfrak{Op}^A(F)$ is bounded by some semi-norm defining the
Fr\'{e}chet topology on $S^0_0(\Xi)$ and this semi-norm only depends on the
dimension $d$ of $\X$ and some Fr\'{e}chet semi-norm of the components of the
magnetic field in $BC^\infty(\X)$ (this second fact, although not explicitly
stated there, easily follows when looking
at the detailed proof of Theorem 3.1 in \cite{IMP1-pprt}). We shall
define the following associated norm on the $S^0_0(\Xi)$ symbols:
\beq
\|F\|_B\ :=\ \|\mathfrak{Op}^A(F)\|_{\mathbb{B}(\mathcal{H})}.
\eeq
In \cite{MP1} it is proven that $\mathfrak{Op}^A(F)$ is Hilbert-Schmidt if and
only if $F\in L^2(\Xi)$ and 
$
{\|F\|_{L^2(\Xi)}=
\|\mathfrak{Op}^A(F)\|_{\mathbb{B}_2(\mathcal{H})}}.
$

\subsection{The main result.}

In the papers \cite{Arsu-JOT,Arsu}, G. Arsu uses some ideas and results of H.O.
Cordes
\cite{Cordes} and T. Kato \cite{Kato} and the characterization of Schatten-von
Neumann classes of operators coming from J.W. Calvin and R. Schatten \cite{Birman,Simon2} 
in order to obtain an interesting
criterion for a Weyl operator to be in a given Schatten-von Neumann class. 
Our aim in this paper is to replace the usual Weyl system with the magnetic
Weyl system \eqref{magn-W-syst} and prove a kind of a similar criterion for a {\it magnetic
Weyl operator} \eqref{magn-W-op}. We
prove the following
Theorem.
\begin{theorem}\label{main-Th}
Suppose that $B$ is a magnetic field with
components of class $BC^\infty(\X)$ and let $A$ be a vector potential 
for $B$. Suppose that
$F\in\mathscr{S}^\prime(\Xi)$ and let us denote
by $s(d):=2[d/2]+2$ and $t(d):=d+[d/2]+1$. If $\partial_x^\alpha\partial_\xi^\beta F\in
L^1(\Xi)$ for 
 $|\alpha|\leq s(d)$ and $|\beta|\leq t(d)$, then
$\mathfrak{Op}^A(F)\in\mathbb{B}_1\big(L^2(\X)\big)$ and there exists some
finite constant $C>0$ such that 
$$
\|\mathfrak{Op}^A(F)\|_{\mathbb{B}_1(\mathcal{H})}\leq C\underset{|\alpha|\leq
s(d)}{\sum}\ \underset{|\beta|\leq
t(d)}{\sum}\left\|\partial_x^\alpha\partial_\xi^\beta
F\right\|_{L^1(\Xi)}.
$$
\end{theorem}
\begin{remark}
We note that this Theorem is the {\it 'magnetic'
version} of the case $p=1$ of Theorem 6.4 in
\cite{Arsu-JOT}. Let us consider the value of $s(d)\in\mathbb{N}$ (the
number of
derivatives with respect to the $\X$-variables) that we obtain. For
$d\in\mathbb{N}$ odd, we have $s(d)=d+1$ exactly as in \cite{Arsu-JOT}, while
for $d\in\mathbb{N}$ even we have $s(d)=d+2$ that is larger by one unit with
respect to the value in \cite{Arsu-JOT}; this is just the consequence of our
choice to work without fractional derivatives, in order not to complicate to much 
the technical arguments.
Concerning $t(d)\in\mathbb{N}$, it is interesting to note that it is larger
then the
value given in \cite{Arsu-JOT} for the zero magnetic field situation and that
reflects the fact that the presence of a magnetic field that does not vanish at
infinity obliges us to control several derivatives of the symbol. Moreover,
if we go into the details of our proof of Theorem \ref{main-Th} (more precisely
the proof of Proposition \ref{Cordes-magn}) we easily see that in the absence
of the magnetic field (i.e. of the factor $\omega^{\mathcal{T}^{\X}_zB}$) we can take
$t(d)=d+1$ as in \cite{Arsu-JOT}. 
\end{remark}
\begin{remark}
Considering the usual Sobolev spaces $W^{m,p}(\Xi)$ on $\Xi$, we notice that our main Theorem \ref{main-Th} implies that $\mathfrak{Op}^A:W^{m(d),1}(\Xi)\rightarrow\mathbb{B}_1(\mathcal{H})$ is a continuous operator for $m(d)=\max\{s(d),t(d)\}$. Going back to the 'magnetic' version of the Calderon-Vaillancourt Theorem (Theorem 3.1 in \cite{IMP1}) we notice that it implies that $\mathfrak{Op}^A:W^{p(d),\infty}(\Xi)\rightarrow\mathbb{B}(\mathcal{H})$ is a bounded operator. The analysis in \cite{T} (p. 147) shows that for $1<p<\infty$ the Schatten-von Neumann class $\mathbb{B}_p(\mathcal{H})$ is an interpolation space for the pair $\big(\mathbb{B}_1(\mathcal{H}),\mathbb{B}(\mathcal{H})\big)$, either for the real interpolation method or for the complex one. Using then Theorem 6.4.5 (points (5) or (7)) in \cite{B-L} we obtain the continuity of the 'magnetic quantization' as operator $\mathfrak{Op}^A:W^{n(d,p),p}(\Xi)\rightarrow\mathbb{B}_p(\mathcal{H})$ for some well-defined $n(d,p)\in\mathbb{R}_+$ and thus the following Corollary of our Theorem \ref{main-Th}.
\end{remark}
\begin{corollary}\label{main-C}
Suppose that $B$ is a magnetic field with
components of class $BC^\infty(\X)$ and let $A$ be a vector potential  for $B$. Then for any
$1<p\leq\infty$ there exists some $n(d,p)\in\mathbb{N}$ such that $\mathfrak{Op}^A:\mathscr{S}(\Xi)\rightarrow\mathbb{B}(\mathcal{H})$ defines a bounded operator from $W^{n(d,p),p}(\Xi)$ to $\mathbb{B}_p(\mathcal{H})$ with a norm  depending on $d$, $p$ and some $BC^\infty$ semi-norm of $B$.
\end{corollary}
\begin{remark}
We notice that the index $n(d,p)$ appearing in the above Corollary is not optimal. We also notice that a somehow related compactness criterion has been given in \cite{I}.
\end{remark}

\section{Proof of Theorem \ref{main-Th}.}\label{S2}

While the idea of the proof follows closely the arguments and some results from
\cite{Arsu-JOT,Cordes,Kato}, several essential technical steps have to be
completely reconsidered in order to control the {\it 'magnetic phase factors'}
present in the magnetic Weyl calculus. 

Let us recall that in \cite{Arsu-JOT,Kato} one begins by noticing that
the fundamental solutions of some simple elliptic differential operators are
symbols of trace-class operators (as implied by Cordes Lemma
\cite{Cordes,Arsu}). Then, starting from  the following formula
valid for two
symbols $f$ and $g$ of class $\mathscr{S}(\Xi)$ 
\beq\label{Kato-1}
\mathfrak{Op}(f*g)=\int_\Xi\,f(X)\mathfrak{Op}(\mathcal{T}^{\Xi}_{X}g)\,
dX=\int_\Xi\,f(X)\Big(W(-X)\mathfrak{Op}(g)W(X)\Big)\,dX,
\eeq
a procedure elaborated by G.
Arsu \cite{Arsu-JOT} using
the results of J.W. Calkin and R. Schatten \cite{Birman,Simon1,Simon2} and some ideas of T. Kato \cite{Kato} allows to
obtain the desired result. Let us develop these ideas and adapt them to
our situation.

For $(s,t)\in\mathbb{R}_+\times\mathbb{R}_+$ let us consider the following
$\Psi$DO on $\Xi$:
\beq
\mathfrak{L}_{s,t}\ :=\
\big(\bb1-\Delta_\X\big)^{s/2}\big(\bb1-\Delta_{\X^*}\big)^{t/2}
\eeq
where
\beq
\Delta_\X\ :=\ \underset{1\leq j\leq d}{\sum}\partial_{x_j}^2,\qquad
\Delta_{\X^*}\ :=\ \underset{1\leq j\leq d}{\sum}\partial_{\xi_j}^2.
\eeq
Let us denote by ${\psi_s}\in\mathscr{S}^\prime(\X)$ the unique
fundamental solution of $\big(\bb1-\Delta_\X\big)^{s/2}$ and by 
${\dot{\psi}_t}\in\mathscr{S}^\prime(\X^*)$ the unique
fundamental solution of $\big(\bb1-\Delta_{\X^*}\big)^{t/2}$.
Let us recall the following well known result (see for example
section 5 in \cite{Arsu-JOT} and Corollary 2.6 in \cite{Arsu} for the last
statement).
\begin{proposition}\label{P-Bes}
For any $s>0$ the distribution ${\psi_s}\in\mathscr{S}^\prime(\X)$ is in
fact a function of class $L^1(\X)$ that is in
$\mathscr{S}(\X\setminus\{0\})$. For
$|x|\searrow0$ we have
that
\beq
\partial_x^\alpha{\psi_s}\,\sim\,\mathscr{O}\big(1+|x|^{s-d-|\alpha|}\big)
,\quad s-d-|\alpha|\ne0,
\eeq
\beq
\partial_x^\alpha{\psi_s}\,\sim\,\mathscr{O}\big(1+\text{\sf ln}|x|^{
-1} \big),\quad s-d-|\alpha|=0.
\eeq
For $s>d$ we have that $\psi_s\in\mathcal{H}^p(\X)$ for any $p<(s/2)$.
We have evidently a similar behaviour for
${\dot{\psi}_t}\in\mathscr{S}^\prime(\X^*)$.
\end{proposition}
This result and the Cordes Lemma \cite{Cordes,Arsu} allow to prove that
$\Psi_{s,t}:={\psi_s}\otimes{\dot{\psi}_t}$ is a tempered distribution on $\Xi$ 
defining a trace-class
operator on $L^2(\X)$. Then, using \ref{Kato-1} and the trivial fact that for any
$f\in\mathscr{S}^\prime(\Xi)$, 
$$
f\,=\,f*\delta_0\,=\,f*\big(\mathfrak{L}_{s,t}\Psi_{s,t}\big)\,=\,\big(\mathfrak{L}_{s,t}f\big)*\Psi_{s,t}
$$
(with $\delta_0$ the Dirac
measure of mass 1 at $0\in\Xi$), {\it Kato's operator calculus} \cite{Kato} and Lemma 4.3 in
\cite{Arsu-JOT} give the desired result in the absence of the magnetic field.
An important difficulty for the case of the 'magnetic' Weyl calculus comes from
the fact that equation \eqref{Kato-1} is no
longer valid for the magnetic Weyl calculus; more precisely we have
\beq
\mathfrak{Op}^A(\mathcal{T}^{\Xi}_{X}g)\ \ne\ W^A(X)^*\mathfrak{Op}^A(g)W^A(X).
\eeq
The following subsection is devoted to the control of this difficulty.

\subsection{Magnetic translations of symbols.}

In Proposition 3.4 in \cite{IMP2} one defines the action of $\Xi$ on the symbols
in $\mathscr{S}^\prime(\Xi)$ by {\it 'magnetic translations'}:
\beq
\Xi\ni
Z\mapsto\mathfrak{T}^B_Z\in\mathbb{B}\big(\mathscr{S}^\prime(\Xi);\mathscr{S}
^\prime(\Xi)\big)
\eeq
as the conjugate action associated to the {\it magnetic Weyl system}:
\beq
\mathfrak{Op}^A\big(\mathfrak{T}^B_{Z}g\big)\ :=\
W^A(Z)^*\mathfrak{Op}^A(g)W^A(Z).
\eeq
Let us denote by $\lambda^A_z(x):=\Lambda^A(x,x+z)$ and notice that formula \eqref{magn-W-syst} may be written as
\beq
W^A(Z)=\lambda^A_zW(Z)=W(Z)\big(\mathcal{T}^{\X}_z\lambda^A_z\big).
\eeq
Thus we can write
\beq
\mathfrak{Op}^A\big(\mathfrak{T}^B_{Z}g\big)=W(-Z)\overline{\lambda^A_z}\mathfrak{Op}^A(g)W(Z)
\big(\mathcal{T}^{\X}_z\lambda^A_z\big)=\overline{\big(\mathcal{T}^{\X}_z\lambda^A_z\big)}
W(-Z)\mathfrak{Op}^A(g)W(Z)\big(\mathcal{T}^{\X}_z\lambda^A_z\big).
\eeq
We notice that
\beq
W(-Z)\mathfrak{Op}^A(g)W(Z)=W(-Z)\big(\mathfrak{Int}\Lambda^A\mathfrak{W}g\big)W(Z)=
\mathfrak{Int}\Lambda^{\mathcal{T}^{\X}_zA}\mathfrak{W}(\mathcal{T}^\Xi_Zg)=\mathfrak{Op}^{\mathcal{T}^{\X}_zA}
\big(\mathcal{T}^\Xi_Zg\big),
\eeq
i.e.
\beq
\mathfrak{Op}^{\mathcal{T}^{\X}_zA}
\big(\mathcal{T}^\Xi_Zg\big)=\big(\mathcal{T}^{\X}_z\lambda^A_z\big)
W^A(Z)^*\mathfrak{Op}^A(g)W^A(Z)\overline{\big(\mathcal{T}^{\X}_z\lambda^A_z\big)}.
\eeq

Finally, replacing $B$ with $\mathcal{T}^\X_{-z}B$ and $A$ with ${\mathcal{T}^\X_{-z}A}$ and denoting by $\mathfrak{U}^A_z$ the unitary operator of multiplication with the modulus 1 function $\lambda^A_z$ we get
\beq
\mathfrak{Op}^A\big(\mathcal{T}^\Xi_Zg\big)\ =\ (\mathfrak{U}^A_z)^*W^{\mathcal{T}^\X_{-z}A}(Z)^*\mathfrak{Op}^{\mathcal{T}^\X_{-z}A}(g)W^{\mathcal{T}^\X_{-z}A}(Z)\mathfrak{U}^A_z
\eeq

These arguments allow us to write
$$
\mathfrak{Op}^A(f*g)\,=\,\int_\Xi\,f(Z)\mathfrak{Op}^A(\mathcal{T}^\Xi_Zg)dZ\,=
$$
\beq\label{Kato-rem-magn}
=\,\int_\Xi\,f(Z)(\mathfrak{U}^A_z)^*W^{\mathcal{T}^\X_{-z}A}(Z)^*\mathfrak{Op}^{\mathcal{T}^\X_{-z}A}(g)W^{\mathcal{T}^\X_{-z}A}(Z)\mathfrak{U}^A_zdZ.
\eeq
This last formula replaces \eqref{Kato-1} in the case of the 'magnetic' Weyl
calculus.

\subsection{Kato's operator calculus.}

We recall here one of the main results in \cite{Arsu-JOT} using the {\it
operator calculus} elaborated by T. Kato in \cite{Kato}. Suppose given \textit{a
map} $V:\Xi\rightarrow\mathbb{B}(\mathcal{H})$ 
measurable for the weak operator topology
on $\mathbb{B}(\mathcal{H})$. For any trace-class
operator $T\in\mathbb{B}_1(\mathcal{H})$ and any $\varphi\in\mathscr{S}(\Xi)$ we
can define the following integral (with respect to the weak operator
topology):
\beq\label{Kato-op}
\varphi\{T\}\ :=\ \int_\Xi\varphi(X)\big(V(X)^*TV(X)\big)\,dX.
\eeq

\begin{proposition}\label{T-Kato} (Point (b) of Lemma 4.3 in \cite{Arsu-JOT})
If there exists a finite $C>0$ such that 
$\|V(X)\|_{\mathbb{B}(\mathcal{H})}\leq\sqrt{C}$ almost everywhere on $\Xi$,
then for any $\varphi\in L^1(\Xi)$ the integral \eqref{Kato-op} is well
defined in the weak operator topology on $\mathbb{B}(\mathcal{H})$, belongs to
$\mathbb{B}_1(\mathcal{H})$ and we have the estimation
\beq\label{est-G-2}
\|\varphi\{T\}\|_{\mathbb{B}_1(\mathcal{H})}\ \leq\
C\|\varphi\|_{L^1(X)}\|T\|_{\mathbb{B}_1(\mathcal{H})}.
\eeq
\end{proposition}

\begin{remark}
We notice that for any vector potential, the map $\Xi\ni Z\mapsto
W^{\mathcal{T}^\X_{-z}A}(Z)\mathfrak{U}^A_z\in\mathbb{B}(\mathcal{H})$ satisfies the condition in Theorem
\ref{T-Kato} with a constant $C=1$, due to their unitarity. Moreover, the proof of point (b) in Lemma 4.3 in \cite{Arsu-JOT} clearly remains true if we replace the trace-class operator $T$ by any function $\X\ni z\mapsto T_z\in\mathbb{B}_1(\mathcal{H})$ with bounded $\mathbb{B}_1(\mathcal{H})$-norm.
\end{remark}
\begin{remark}\label{Rem-p}
Let us notice that if one wants to treat the case $p>1$ for the magnetic quantization in a way similar to points 
(a) and (c) in Lemma 4.3 in 
\cite{Arsu-JOT}, an important difficulty comes from the fact that, due to formula \eqref{Kato-rem-magn}, 
one has to consider a function $\X\ni z\mapsto T_z\in\mathbb{B}_1(\mathcal{H})$ instead of a constant factor $T\in\mathbb{B}_1(\mathcal{H})$. Some simple examples show that the proof of point (a) in Lemma 4.3 in \cite{Arsu-JOT} is no longer valid in this situation.
\end{remark}

Using \eqref{Kato-rem-magn} and the above Remark we obtain the following
Corollary of Proposition \ref{T-Kato} (the {\it 'magnetic version'} of the case $p=1$ in Theorem
4.5 in \cite{Arsu-JOT}):

\begin{corollary}\label{cor-Kato}
Suppose given a magnetic field $B$ with
components of class $BC^\infty(\X)$ and suppose fixed some vector potential $A$
for $B$; if a symbol
$F\in\mathscr{S}^\prime(\Xi)$ has the property
$\mathfrak{Op}^A(F)\in\mathbb{B}_1(\mathcal{H})$, then for any ${f\in L^1(\Xi)}$ we have that
${\mathfrak{Op}^A(f*F)\in\mathbb{B}_1(\mathcal{H})}$ and
$${\|\mathfrak{Op}^A(f*F)\|_{\mathbb{B}_1(\mathcal{H})}\leq\|f\|_{L^1(\Xi)}
\|\mathfrak{Op}^A(F)\|_{\mathbb{B}_1(\mathcal{H})}}.
$$
\end{corollary}

\subsection{Trace-class property of $\mathfrak{Op}^{A}(\Psi_{s,t})$.}

To finish the proof of our Theorem \ref{main-Th} it is enough to prove the following {\it 'magnetic version'} of Lemma 1 in
\cite{Cordes}. 

\begin{proposition}\label{Cordes-magn} Suppose given a magnetic field $B=dA$
with components of class $BC^\infty(\X)$;
for $t>3d/2$ and $s>2[d/2]+2$ we have that
$\mathfrak{Op}^{\mathcal{T}^\X_{-z}A}(\Psi_{s,t})\in\mathbb{B}_1(\mathcal{H})$ uniformly for
$z\in\X$.
\end{proposition}
\begin{proof}
We shall proceed as in \cite{Cordes,Arsu} but we shall work with the
magnetic Moyal product \eqref{magn-Moyal-prod}.
The idea is to write
$\Psi_{s,t}$ as a magnetic Moyal product of two symbols of class $L^2(\Xi)$:
\beq
\Psi_{s,t}\ =\
\Phi^{(1)}\sharp^{\mathcal{T}^\X_{-z}B}\Phi^{(2)},\qquad\Phi^{(j)}\in
L^2(\Xi),\ j=1,2.
\eeq
Let us use the following shorthand notations for the magnetic Moyal product with a translated magnetic field and the corresponding magnetic inverse:
\beq
\sharp^B_z\ :=\ \sharp^{\mathcal{T}^\X_{-z}B};\qquad F^-_{B,z}\ :=\ F^-_{\mathcal{T}^\X_{-z}B}.
\eeq

Let us consider the symbols
$p_{m,\lambda}(X):=<\xi>^m+\lambda$ for any $m>0$ and some $\lambda>0$ large
enough; they are evidently elliptic symbols of class $S^m_1(\Xi)$ that, for
$\lambda>0$ large enough, are
invertible for the {\it magnetic Moyal product} due to Theorem 1.8 in
\cite{MPR1}. More precisely, looking at the proof of this cited Theorem we
see that 
\beq
r_{m,\lambda}\,:=\,\big(p_{m,\lambda}\big)^-_{B,z}\,=\,
\big(<\xi>^m+\lambda\big)^{-1}{\sharp^B_z}
\left(\underset{k\in\mathbb{N}}{\sum}\quad\underset{k}{\underbrace{s_{z,m}
(\lambda){\sharp^B_z}\ldots{\sharp^B_z}
{s_{z,m}(\lambda)}}}
\right)
\eeq
with $s_{z,m}(\lambda)\in S^{-\kappa}_1(\Xi)$ for some $\kappa\in(0,1)$, having the
operator norm strictly less then 1 for $\lambda>0$ large enough and
the defining Fr\'{e}chet semi-norms bounded by some semi-norm of the components of
$\mathcal{T}^\X_{-z}B$
in $BC^\infty(\X)$; as these semi-norms are translation invariant, we have
uniform bounds for $z\in\X$. Thus,
using Proposition \ref{est-m-M-prod} in the Appendix and Proposition 6.2 in
\cite{IMP2} we conclude that for $\lambda>0$ large enough, the symbol semi-norms
of
$r_{m,\lambda}\in S^{-m}_1(\Xi)$ are bounded by some
constants
that do not depend on $z\in\X$.

Let us also consider the function $q_r(X):=<x>^r$ with $r\in\mathbb{R}$,
defining a symbol of class 
$S^0_1(\Xi)$ for any $r\leq0$. Formally
we can write
\beq\label{HS-prod}
\Psi_{s,t}\ =\
\big(q_{-r}{\sharp^B_z}r_{m,\lambda}\big){\sharp^B_z}
\big(p_{m,\lambda}{\sharp^B_z}q_r{\sharp^B_z}\Psi_{s,t}\big)
\eeq
Using once again Proposition \ref{est-m-M-prod} in the Appendix and the fact
that the semi-norms of the components of the magnetic field that control the
magnetic Moyal products are translation
invariant, we easily conclude that for
$r>0$, $m>0$ and $\lambda>0$ large enough
\beq
\big(q_{-r}{\sharp^B_z}r_{m,\lambda}\big)\ \in\ S^{-m}_1(\Xi),
\eeq
uniformly for $z\in\X$. Moreover, for any $a\geq0$
and $b\geq0$ we can write:
\beq
<x>^a<\xi>^b\big(q_{-r}{\sharp^B_z}r_{m,\lambda}\big)(x,\xi)\,
=
\eeq
$$
=\,\pi^{-2d}<x>^a<\xi>^b\int\limits_{\Xi\times\Xi}e^{-2i\sigma(Y,Y^\prime)}
\omega^{
\mathcal{T}^X_{-z}B}(x,y,y^\prime)<x-y>^{-r}r_{m,\lambda}(x-y^\prime,\xi-\eta^\prime)
dY\,dY^\prime\,=
$$
$$
=\,\pi^{-d}C_a\int_\X<x-y>^{-(r-a)}\left(\frac{<x>^a}{<x-y>^a<y>^a}
\right)\times
$$
$$
\times\left[\int_{\X^*}<y>^a<\eta^\prime>^be^{2i<\eta^\prime,y>}
\left(\frac{<\xi>^b}{<\xi-\eta^\prime>^b<\eta^\prime>^b}
\right)\left(<\xi-\eta^\prime>^b
r_{m,\lambda}(x,\xi-\eta^\prime)\right)d\eta^\prime\right]dy.
$$
We use the identities:
\beq
<y>^{2N_1}e^{2i<\eta^\prime,y>}=(1-4^{-1}\Delta_{\eta^\prime})^{N_1}e^{
2i<\eta^\prime,y>},\qquad
<\eta^\prime>^{2N_2}e^{2i<\eta^\prime,y>}=(1-4^{-1}\Delta_{y})^{N_2}e^{
2i<\eta^\prime,y>}
\eeq
and
after some integrations by parts as in the proof of Proposition
\ref{est-m-M-prod} in the Appendix, taking $0\leq a\leq r$, $0\leq b\leq m$ and
$2N_1\geq[a]+d+1$, $2N_2\geq[b]+d+1$ we get that
\beq
<x>^a<\xi>^b\left|\big(q_{-r}{\sharp^B_z}r_{m,\lambda}\big)(x,\xi)\right|\,
\leq\, C_{a,d}
\underset{(x,\xi)\in\Xi}{\sup}<\xi>^b\underset{|\alpha|\leq
2N_2}{\sum}\left|\big(\partial_\xi^\alpha r_{m,\lambda}\big)(x,\xi)\right|\,\leq
C(a,b)\nu^m_{0,2N_2}(r_{m,\lambda}).
\eeq
A similar computation
can be made for any derivative $\partial_x^\alpha\partial_\xi^\beta
\big(q_{-r}{\sharp^B_z}r_{m,\lambda}\big)$
so that we conclude that
\beq
q_a\,p_{b,0}\big(q_{-r}{\sharp^B_z}r_{m,\lambda}\big)\in
S^0_1(\Xi),\qquad\forall(a,b)\in[0,r]\times[0,m]
\eeq
uniformly in $z\in\X$ and taking $r>d/2$ and $m>d/2$ we note that
$\Phi^{(1)}:=q_{-r}{\sharp^B_z}r_{m,\lambda}\in L^2(\Xi)$ so that
$\mathfrak{Op}^{\mathcal{T}^\X_{-z}A}(\Phi^{(1)})\in\mathbb{B}_2\big(L^2(\X)\big)$ uniformly in $z\in\X$.

Now let us study the second factor in \eqref{HS-prod}. We note that for $m>0$
and $r>0$
the first two functions of this second magnetic Moyal product, 
namely $p_{m,\lambda}$
and $q_r$, are in fact $C^\infty(\Xi)$ functions with
polynomial growth at infinity uniformly for all their derivatives, and thus
Proposition 4.23 in \cite{MP1} shows
that their magnetic Moyal product may be well defined in the sense of tempered
distributions and moreover this product (as a tempered distribution) may be
further composed by magnetic Moyal product with any tempered distribution on
$\Xi$. Thus $\Phi^{(2)}$ is well defined as a tempered distribution on $\Xi$ and
we can also use the associativity of the magnetic Moyal product. Let us note
that this tempered distribution depends in fact on $z\in\X$ due to the
translated magnetic field appearing in the two 'magnetic' Moyal products in the
definition of $\Phi^{(2)}$ as the second parenthesis in \eqref{HS-prod}; thus
we shall use the notation $\Phi^{(2)}_z$ and notice that this dependence is
uniformly smooth with respect to the weak distribution topology.

We begin by computing
$q_r{\sharp^B_z}\Psi_{s,t}=q_r{\sharp^B_z}\big(\psi_s\otimes\dot{\psi}
_t\big)$ for $r>d/2>0$:
$$
\big[q_r{\sharp^B_z}\big(\psi_s\otimes\dot{\psi}
_t\big)\big](x,\xi)\ =
$$
$$
=\ \pi^{-2d}\int\limits_{\Xi\times\Xi}e^{-2i\sigma(Y,Y^\prime)}
\omega^{\mathcal{T}^\X_{-z}B}(x,y,y^\prime)<x-y>^{r}\psi_s(x-y^\prime)\dot{\psi}
_t(\xi-\eta^\prime)dY\,dY^\prime\ =
$$
\beq
=\
\pi^{-d}\psi_s(x)\int\limits_{\X\times\X^*}e^{2i<\eta^\prime,y>}<x-y>^{r}
\dot{\psi}_t(\xi-\eta^\prime)dy\,d\eta^\prime\ =
\eeq
\beq
=\ \pi^{-d}\psi_s(x)\int\limits_\X e^{2i<\xi,y>}<x-y>^{r}\left(\int_{\X^*}
e^{-2i<\xi-\eta^\prime,y>}\dot{\psi}_t(\xi-\eta^\prime)d\eta^\prime\right)dy=
\eeq
\beq
=\ 2^d\psi_s(x)\int\limits_\X e^{2i<\xi,y>}\frac{<x-y>^{r}}{<2y>^t}dy\ =\
2^d\big(q_r\psi_s\big)(x)\int\limits_\X
e^{2i<\xi,y>}\frac{<x-y>^{r}}{<x>^r<2y>^t}dy\ =
\eeq
\beq
=\ (2\pi)^{d/2}\big(q_r\psi_s\big)(x)\left((\bb1\otimes\mathcal{F}^-_{\X})f_{r,t}\right)(x,\xi)
\eeq
where:
\beq\label{def-f}
f_{r,t}(x,y)\ :=\ \frac{<x-(y/2)>^{r}}{<x>^r<y>^t}.
\eeq
It is easy to check that $f_{r,t}\in C^\infty_{\text{\sf pol}}(\X\times\X)$ and
satisfies the estimations:
\beq\label{est-f}
\left|\big(\partial^\alpha_x\partial^\beta_yf_{r,t}\big)(x,y)\right|\ \leq\
C_{\alpha\beta}<x>^{-|\alpha|}<y>^{r-t-|\beta|}.
\eeq

Now let us consider some $m>d/2$, and use the notations:
$\widetilde{f}_{r,t}:=(2\pi)^{d/2}(\bb1\otimes\mathcal{F}^-_{\X})f_{r,t}$ and
for any $r\geq0$ the function $\widetilde{\psi}_{s,r}(x):=<x>^r\psi_s(x)$. We
notice that for any $r\geq0$ the function
$\widetilde{\psi}_{s,r}$ has exactly the same
properties as those of $\psi_s$ given in Proposition \ref{P-Bes}.

We want to show that:
\beq\label{6.12}
\Phi^{(2)}_z\,:=\,p_{m,\lambda}{\sharp^B_z}q_r{\sharp^B_z}\Psi_{s,t}
\,=\,
(2\pi)^{d/2}\Big(p_{m,\lambda}{\sharp^B_z}\left[
\big(q_r\psi_s\otimes1\big)\left((\bb1\otimes\mathcal{F}^-_{\X})f_{r,t}\right)\right]
\Big)
\eeq
as a tempered distribution on $\Xi$ is in fact an $L^2(\Xi)$ function uniformly
for $z\in\X$. In order
to deal with the possible singularities of this distribution we shall
regularize it by introducing 4 cut-off functions in the oscillatory
integrals appearing in the definition \eqref{magn-Moyal-prod}, more precisely
we shall approach $\Phi^{(2)}_z$, in the weak distribution topology, by the
following continuous functions on $\Xi$ depending also on 4 positive parameters
$\{R_j\}_{j=1,2,3,4}$:
\beq
\widetilde{\Phi^{(2)}}_{(R_j,z)}(x,\xi)\,:=\,
\int\limits_\Xi\int\limits_\Xi\,e^{-2i<\eta,y^\prime>}e^{2i<\eta^\prime,y>}
\chi_{R_1}(y)\chi_{R_2}(y^\prime)\chi_{R_3}(\eta)\chi_{R_4}(\eta^\prime)\,\times
\eeq
$$
\times\,p_{m,\lambda}(\xi-\eta)\,\widetilde{\psi}_{s,r}(x-y^\prime)\,\widetilde{
f}_{r,t} (x-y^\prime ,
\xi-\eta^\prime)\,\omega^{\mathcal{T}^\X_{-z}B}(x,y,y^\prime)\,dy\,dy^\prime\,d\eta\,
d\eta^\prime,
$$
where for any $R>0$ we define $\chi_R(v):=\chi(R^{-1}|v|)$ with
$\chi:\mathbb{R}_+\rightarrow\mathbb{R}_+$ a smooth decreasing
function that satisfies $\chi(t)=1$ for $0\leq t\leq1$ and $\chi(t)=0$ for
$t\geq2$.

We shall first consider the term $<\xi-\eta>^m$ in the function
$p_{m,\lambda}(\xi-\eta)=<\xi-\eta>^m+\lambda$ and the associated integral
\beq\label{11.12-1}
\widetilde{\Phi^{(3)}}_{(R_j,z)}(x,\xi)\,:=\,
\int\limits_\Xi\int\limits_\Xi\,e^{-2i<\eta,y^\prime>}e^{2i<\eta^\prime,y>}
\chi_{R_1}(y)\chi_{R_2}(y^\prime)\chi_{R_3}(\eta)\chi_{R_4}(\eta^\prime)\,\times
\eeq
$$
\times\,<\xi-\eta>^m\,\widetilde{\psi}_{s,r}(x-y^\prime)\,\widetilde{
f}_{r,t} (x-y^\prime ,
\xi-\eta^\prime)\,\omega^{\mathcal{T}^\X_{-z}B}(x,y,y^\prime)\,dy\,dy^\prime\,d\eta\,
d\eta^\prime.
$$
We make the measure preserving change of variables:
\beq
(y,y^\prime,\eta,\eta^\prime)\mapsto(y,u,\zeta,\zeta^\prime);\qquad
\left\{
\begin{array}{l}
 u:=x-y^\prime\\
 \zeta:=\xi-\eta\\
 \zeta^\prime:=\xi-\eta^\prime,
\end{array}
\right.
\eeq
so that \eqref{11.12-1} may be written as
\beq
\widetilde{\Phi^{(3)}}_{(R_j,z)}(x,\xi)\,:=\,
\int\limits_\Xi\int\limits_\Xi\,e^{-2i<\xi,x-u-y>}e^{2i<\zeta,x-u>}e^{
-2i<\zeta^\prime,y>}
\chi_{R_1}(y)\chi_{R_2}(x-u)\chi_{R_3}(\xi-\zeta)\chi_{R_4}(\xi-\zeta^\prime)\,
\times
\eeq
$$
\times\,<\zeta>^m\,\widetilde{\psi}_{s,r}(u)\,\widetilde{f}_{r,t}(u,\zeta^\prime)\,
\omega^{\mathcal{T}^\X_{-z}B}(x,y,x-u)\,dy\,du\,d\zeta\,
d\zeta^\prime\,=
$$
\beq\label{Phi-3-2}
=\,\int_{\X^*}e^{2i<\zeta,x>}\,<\zeta>^m\,\chi_{R_3}(\xi-\zeta)\left\{
\int_\X\,e^{-2i<\zeta,u>}\,\widetilde{\psi}_{s,r}(u)\,\chi_{R_2}(x-u)\,
\right.\,\times
\eeq
$$
\times\,\left.\left[\int_\X\,e^{-2i<\xi,x-u-y>}\left(\int_{\X^*}e^{
-2i<\zeta^\prime,y>}\chi_{R_4}
(\xi-\zeta^\prime)\widetilde{
f}_{r,t}(u,\zeta^\prime)\,d\zeta^\prime\right)\,\omega^{\mathcal{T}^\X_{-z}B}(x,y,x-u)\chi_{
R_1}(y)\, dy\right]du\right\}d\zeta\,\equiv
$$
\beq\label{f1-ian}
\equiv\,\int_{\X^*}e^{2i<\zeta,x>}\,<\zeta>^m\,\chi_{R_3}(\xi-\zeta)\left(
\int_\X\,e^{-2i<\zeta,u>}\,\Theta_{(R_j,z)}(x,\xi,u)du\right)d\zeta.
\eeq
Let us study closer the continuous function introduced in \eqref{f1-ian}:
\beq
\Theta_{(R_j,z)}(x,\xi,u)\,:=\,\widetilde{\psi}_{s,r}(u)\,\chi_{R_2}(x-u)\,
\times
\eeq
$$
\times\,\left[
\int_\X\,e^{-2i<\xi,x-u-y>}\left(\int_{\X^*}e^{-2i<\zeta^\prime,y>}\chi_{R_4}
(\xi-\zeta^\prime)\widetilde{
f}_{r,t}(u,\zeta^\prime)\,d\zeta^\prime\right)\,\omega^{\mathcal{T}^\X_{-z}B}(x,y,x-u)\chi_{
R_1}(y)\, dy\right].
$$
We make the change of variable $\X\ni y\mapsto v:=x-u-y\in\X$ that allow
us to write it as:
\beq\label{Theta-1}
\Theta_{(R_j,z)}(x,\xi,u)\,:=\,\widetilde{\psi}_{s,r}(u)\,\chi_{R_2}(x-u)\,
T_{R_1,R_4,z}(x,\xi,u),
\eeq
\beq
T_{R_1,R_4,z}(x,\xi,u)\,:=
\eeq
$$
=\,\left[
\int_\X e^{-2i<\xi,v>}\left(\int_{\X^*}e^{-2i<\zeta^\prime,x-u-v>}\chi_{R_4}
(\xi-\zeta^\prime)\widetilde{
f}_{r,t}(u,\zeta^\prime)\,d\zeta^\prime\right)\omega^{\mathcal{T}^\X_{-z}B}(x,x-u-v,x-u)\chi_{
R_1}(x-u-v)\,dv\right].
$$
We recall that
\beq
\widetilde{f}_{r,t}\,:=\,(2\pi)^{d/2}(\bb1\otimes\mathcal{F}^-_{
\X})f_{r,t}
\eeq
and the fact that the distribution $f_{r,t}\in\mathscr{S}^\prime(\X\times\X)$
defined in \eqref{def-f} is in fact a smooth function. 
Moreover we have that the function
\beq
\X\times\X\ni(u,v)\mapsto<v>^{t-r}f_{r,t}(u,v)\in\mathbb{C}
\eeq
is of class $BC^\infty(\X\times\X)$ so that for $t>r+(d/2)>d$, $f_{r,t}$
belongs to $BC\big(\X_u;L^2(\X_v)\big)$  (sometimes we indicate with an index the variable in $\X$). Using the Fourier inversion Theorem
and noticing that for any $g\in BC\big(\X;L^2(\X)\big)$ we have that
$\|(\bb1\otimes\tau_{-u})g(u,\cdot)\|_{L^2(\X)}=\|g(u,\cdot)\|_{L^2(\X)}$, we
conclude that the tempered distributions
\beq
T_{R_4}(x,\xi,u,v)\,:=\,\int_{\X^*}e^{-2i<\zeta^\prime,x-u-v>}\chi_{R_4}
(\xi-\zeta^\prime)\widetilde{
f}_{r,t}(u,\zeta^\prime)\,d\zeta^\prime,\qquad R_4\in[1,\infty)
\eeq
are a family of functions of class $BC\big(\X_u;L^2(\X_v)\big)$ and by the definition of the Fourier transform on $L^2(\X)$:
\beq
\forall(x,\xi)\in\Xi,\quad\exists\underset{R_4\nearrow\infty}{\lim}T_{
R_4 }
(x,\xi,u,v)\,=\,(2\pi)^df_{r,t}(u,2(x-u-v)),\ \text{in }BC\big(\X_u;L^2(\X_v)\big)
\eeq
uniformly with respect to $(x,\xi)\in\Xi$. Due to the fact that by
definition we have that $\omega^{\mathcal{T}^\X_{-z}B}\in BC(\X^3)$ uniformly and smoothly
for $z\in\X$ we conclude that
\beq
\forall(z,x,\xi)\in\X\times\Xi,\quad\exists\underset{R_4\nearrow\infty
}{\lim}T_{R_4}(x,\xi,u,v)\omega^{\mathcal{T}^\X_{-z}B}(x,x-u-v,x-u)\chi_{R_1}(x-u-v)
\,=
\eeq
$$
=\,(2\pi)^df_{r,t}(u,2(x-u-v))\omega^{\mathcal{T}^\X_{-z}B}(x,x-u-v,x-u)\chi_{
R_1}(x-u-v)=:\theta^{B}_{R_1,z}(x,u,v),
$$
in $BC\big(\X_u;L^2(\X_v)\big)$ uniformly with respect to
$(z,x,\xi,R_1)\in\X\times\Xi\times\mathbb{R}_+$. Moreover, for any magnetic field $B$ with components
of class $BC^\infty(\X)$ we have that for any $(x,u)\in\X^2$
\beq
\left|\theta^B_{R_1,z}(x,u,v)\right|\leq
Cf_{r,t}(u,2(x-u-v)=C<2(x-u-v)>^{r-t}
\eeq
and thus 
\beq
\underset{(x,u)\in\X^2}{\sup}\left\|\theta^B_{R_1,z}(x,u,\cdot)\right\|_{L^2(\X)}\leq
C\|q_{r-t}\|_{L^2(\X)}.
\eeq
We easily conclude that $\theta^B_{R_1,z}\in BC(\X^3)\cap
BC\big(\X_x\times\X_u;L^2(\X_v)\big)$ and the map $\X\ni
z\mapsto\theta^{\mathcal{T}^\X_{-z}B}\in BC\big(\X_x\times\X_u;L^2(\X_v)\big)$ is smooth
and bounded.
Using Plancherel Theorem and the Dominated Convergence Theorem we conclude that
\beq
\forall z\in\X,\quad\exists\underset{R_1\rightarrow\infty}{\lim}
\underset{R_4\rightarrow\infty}{\lim}
T_{R_1,R_4,z}\,=:\,F_z,\quad\text{in
}BC\big(\X_x\times\X_u;L^2(\X^*_\xi)\big),
\eeq
uniformly with respect to $z\in\X$. 

Finally, noticing that by
Proposition \ref{P-Bes}, $\widetilde{\psi}_{s,r}\in L^2(\X)$ for any
$(s,r)\in\mathbb{R}_+\times\mathbb{R}_+$ we conclude that
\beq
\forall z\in\X,\qquad\exists\underset{R_1\rightarrow\infty}{\lim}
\underset{R_2\rightarrow\infty}{\lim}\underset{R_4\rightarrow\infty}{\lim}
\Theta_{(R_j,z)}\,=\,
\big(1\otimes1\otimes\widetilde{\psi}_{s,r}\big)F_{z},\qquad\text{
in }L^2\big(\X_u;BC\big(\X_x;L^2(\X^*_\xi)\big)\big)
\eeq
uniformly for $z\in\X$.

In order to control the factor $<\zeta>^m$ in the first integral in
\eqref{Phi-3-2}, that we consider as a Fourier transform of a tempered
distribution, let us study now the derivatives of $\Theta_{(R_j,z)}$ with
respect to the variable $u\in\X$:
\beq
\left(\partial^\alpha_u\Theta_{(R_j,z)}\right)(x,\xi,u),
\qquad|\alpha|=p\in\mathbb{N}^*.
\eeq
By Proposition \ref{P-Bes} we know that for $s>d$ we have that
$\widetilde{\psi}_{s,r}\in\mathcal{H}^p(\X)$ for $p<(s/2)$ and thus all the
derivatives $\partial^\alpha\widetilde{\psi}_{s,r}$ are of class $L^2(\X)$ for
$s>2|\alpha|$. Let us study the behaviour
of the
distributions
\beq
\partial^\alpha_uF_{z}(x,\xi,u),\qquad|\alpha|=p\in\mathbb{N}^*.
\eeq
When computing $\partial^\alpha_uT_{R_1,R_4,z}$, using Leibniz rule we have to
control the derivatives of order up to $p\in\mathbb{N}$ with respect to $u\in\X$
of $f_{r,t}(u,2(x-u-v))$, of $\omega^{\mathcal{T}^\X_{-z}B}(x,x-u-v,x-u)$ and of the cut-off
functions. Now,
$\partial^\alpha_uf_{r,t}(u,2(x-u-v))$ is easy to compute and due to the estimations \eqref{est-f}, for any
$p\in\mathbb{N}$ these functions have the same properties as the function $f_{r,t}$
in \eqref{def-f}. Using then Lemma 1.1 in
\cite{IMP1} we know that we have the estimations
\beq
\big(\partial^\alpha_y\partial^\beta_{y^\prime}\omega^{\mathcal{T}^\X_{-z}B}\big)(x,y,
y^\prime)\,=\,\theta^{B}_{\alpha,\beta,z}(x,y,
y^\prime)\big(<x>+<y>+<y^\prime>\big)^{|\alpha|+|\beta|}
\eeq
where $\theta^{B}_{\alpha,\beta,z}\in BC(\X^3)$ uniformly in $z\in\X$.
In conclusion we can write
$
\partial^\alpha_u\omega^{\mathcal{T}^\X_{-z}B}(x,x-u-v,x-u)
$
as a finite sum of terms of the form
$
\theta^{B}_z(x,u,v)<x>^p<u>^p<x-u-v>^p
$ with $\theta^{B}_z\in BC(\X^3)$ uniformly in $z\in\X$. We get rid of
the growing factor $<u>^p$ by replacing $\widetilde{\psi}_{s,r}$ by
$\widetilde{\psi}_{s,r+p}$ that has the same properties as
$\widetilde{\psi}_{s,r}$. The
factor $<x-u-v>^p$ may be absorbed in the factor $f_{r,t}$ without changing its
properties that we used above, as long as $t>p+r+(d/2)$.
We remain with the factor $<x>^p$; in order to control its growth at infinity
we turn back at formula \eqref{Phi-3-2} and notice that
\beq
\widetilde{\Phi^{(3)}}_{R_j,z}(x,\xi)\,=
\eeq
$$
=\,\int_{\X^*}\left(\frac{(1-\Delta_\zeta)^{p/2}}{
<2x>^p}e^{2i<\zeta,x>}\right)<\zeta>^m\,\chi_{R_3}(\xi-\zeta)\left\{
\int_\X\,e^{-2i<\zeta,u>}\,\widetilde{\psi}_{s,r}(u)\,\chi_{R_2}(x-u)\,
\Theta_{(R_j,z)}(x,\xi,u)du\right\}d\zeta.
$$
Considering the $\zeta-$integral in the sense of distributions we can transfer
the differential operator $(1-\Delta_\zeta)^{p/2}$ on the $\mathscr{S}(\X^*)$
function
\beq
\X^*\ni\zeta\mapsto<\zeta>^m\,\chi_{R_3}(\xi-\zeta)\left\{
\int_\X\,e^{-2i<\zeta,u>}\,\widetilde{\psi}_{s,r}(u)\,\chi_{R_2}(x-u)\,
\Theta_{(R_j,z)}(x,\xi,u)\right\}\in\mathbb{C}.
\eeq
Using the well known facts that $(1-\Delta)^{-1/2}$ and
$(1-\Delta)^{-1/2}\partial_j$ are bounded operators in $L^2(\X^*)$ we notice
that for $p\in\mathbb{N}$ we can write:
\beq\label{form-develop}
(1-\Delta_\zeta)^{p/2}\,=\,\underset{|\alpha|\leq
p}{\sum}X_\alpha\partial^\alpha_\zeta
\eeq
with $X_\alpha\in\mathbb{B}\big(L^2(\X^*)\big)$ for any $\alpha\in\mathbb{N}^d$.
Then we only have to notice that $\partial^\alpha_\zeta<\zeta>^m$ is a symbol
of type $S^m(\X^*)$ for any $\alpha\in\mathbb{N}^d$ and 
$$
\partial^\alpha_\zeta e^{-2i<\zeta,u>}=(-2i)^{|\alpha|}u^\alpha e^{-2i<\zeta,u>}
$$
and we can control the factor $u^\alpha$ by $<u>^{|\alpha|}$ that can be
absorbed in $\widetilde{\psi}_{s,r}$ for any $|\alpha|\in\mathbb{N}$ without
changing its properties needed for the arguments above to hold. 
Finally we notice that all the terms containing derivatives of the cut-off
functions $\chi_{R_j}$ clearly go to 0 when $R_j\rightarrow\infty$ by the
Lebesgue Dominated Convergence Theorem. In conclusion, for $s>2m$, all the
derivatives
$\partial^\alpha_u\Theta_{(R_j,z)}$ are functions of class
$L^2\big(\X_u;BC\big(\X_x;L^2(\X^*_\xi)\big)\big)$ uniformly for
$z\in\X$ and choosing $m=[d/2]+1$, in the first integral in
\eqref{Phi-3-2} considered as a Fourier transform of a tempered distribution, we
intertwine the multiplication with $<\zeta>^m$ with the
Fourier transform with respect to the variable $u\in\X$. We use formula
\eqref{form-develop} once again and the
Plancherel Theorem noticing that for any $F\in
L^2\big(\X_u;BC\big(\X_x;L^2(\X^*_\xi)\big)\big)$, with $\||F\||^2:=
\int_\X\underset{x\in\X
}{\sup}\int_{\X^*}\left|F(x,\xi,u)\right|^2d\xi\,du$ we have that
\beq
\int_\X\int_{\X^*}\left|F(x,\xi,x)\right|^2d\xi\,dx\leq\int_\X\underset{y\in\X
}{\sup}\int_{\X^*}\left|F(y,\xi,x)\right|^2d\xi\,dx=\||F\|^2.
\eeq
This proves that our distribution
$\widetilde{\Phi^{(3)}}_{R_j,z}$ is in fact a function of class $L^2(\Xi)$
uniformly for $z\in\X$.

The term with $\lambda>0$ replacing $<\xi-\eta>^m$ will also define a
function of class $L^2(\Xi)$ evidently.
The uniformity with respect to $z\in\X$ follows directly from the above remarks
concerning the translation invariance of the bounds. Summarizing we must have:
\beq
r>d/2,\ m>d/2,\ p=m=[d/2]+1,\ t>r+p+d/2>3d/2,\ s>2m=
\left\{\begin{array}{l}
               d+1,\text{ if }d=2p\\
               d+2,\text{ if }d=2p+1.
              \end{array}
\right.
\eeq
\end{proof}

Putting now together Corollary \ref{cor-Kato} and
Proposition \ref{Cordes-magn} and noticing that $t(d)>3d/2$ we obtain the
following result.

\begin{theorem}\label{main-Th-prime}
Suppose given a magnetic field $B$ with
components of class $BC^\infty(\X)$ and suppose fixed some vector potential $A$
for $B$. For
$s\geq s(d)$, $t\geq t(d)$ and $f\in\mathscr{S}^\prime(\Xi)$, if $\mathfrak{L}_{s,t}f\in L^1(\Xi)$, then
$\mathfrak{Op}^A(f)\in\mathbb{B}_1(\mathcal{H})$ and
$\|\mathfrak{Op}^A(f)\|_{\mathbb{B}_1(\mathcal{H})}\leq
C\|\mathfrak{L}_{s,t}f\|_{
L^1(\Xi)}$.
\end{theorem}

This result evidently implies Theorem \ref{main-Th}.

\section{Appendix}

In this Appendix we prove a simplified version of Theorem 2.2 in \cite{IMP1},
that is enough for our analysis in this paper. Not only that in this special
case the proof is much simpler then the one in \cite{IMP1} but we also put into
evidence the dependence on the magnetic field.

\begin{proposition}\label{est-m-M-prod}
For a magnetic field $B$ with components of class $BC^\infty(\X)$ the 'magnetic'
Moyal product
$$
S^m_1(\Xi)\times S^p_1(\Xi)\ni(f,g)\mapsto f\sharp^Bg\in S^{m+p}_1(\Xi)
$$
is continuous for the Fr\'{e}chet topologies being equicontinuous for
$B_{jk}\in\mathcal{B},\forall(j,k)$, with $\mathcal{B}\subset BC^\infty(\X)$ any
bounded subset for its Fr\'{e}chet topology.
\end{proposition}

\begin{proof}
Via a standard cut-off procedure it is enough to consider
$(f,g)\in\mathscr{S}(\Xi)\times\mathscr{S}(\Xi)$ and to prove that there exist
some finite constants $C_{M,N}>0$ and some natural numbers $m_1,m_2,n_1,n_2$ depending on $m,p,M,N$ such that
\beq
\nu^{m+p-N}_{M,N}\big(f\sharp^Bg\big)\leq\,C_{M,N}
\nu_{M-1}(B)\nu^{m-n_1}_{m_1,n_1}(f)\nu^{p-n_2}_{m_2,n_2}(g)
\eeq
where we have considered the semi-norms indexed by $n\in\mathbb{N}$:
\beq
\nu_n(F):=\underset{x\in\X}{\sup}\,\underset{|\alpha|\leq
n}{\sup}\left|\big(\partial_x^\alpha F\big)(x)\right|
\eeq
defining the Fr\'{e}chet topology on $BC^\infty(\X)$ and
\beq
\nu_n(B):=\underset{j,k}{\max}\,\nu_n\big(B_{jk}\big).
\eeq
Thus let us compute
\beq
<\xi>^{-(m+p-|\beta|)}\left(\partial_x^\alpha\partial_\xi^\beta(f\sharp^Bg)
\right)(x,\xi):=
\eeq
$$
\pi^{-2d}<\xi>^{-(m+p-|\beta|)}\partial_x^\alpha\partial_\xi^\beta\left(
\int\limits_{\Xi\times\Xi}e^{-2i\sigma(Y,Z)}
\omega^B(x,y,z)
f(x-y,\xi-\eta)g(x-z,\xi-\zeta)dY\,dZ\right),
$$
that is a finite linear combination of terms of the form
$$
\int\limits_{\Xi\times\Xi}e^{-2i\sigma(Y,Z)}
\big(\partial_x^{\alpha_3}\omega^B(x,y,z)\big)
<\xi>^{-(m-|\beta_1|)}\big(\partial_x^{\alpha_1}\partial_\xi^{\beta_1}f\big)(x-y
,\xi-\eta)<\xi>^{-(p-|\beta_2|)}\big(\partial_x^{\alpha_2}\partial_\xi^{\beta_2}
g\big)(x-z,\xi-\zeta)dY\,dZ
$$
with $\alpha_1+\alpha_2+\alpha_3=\alpha$ and $\beta_1+\beta_2=\beta$.

In order to estimate these integrals we insert the integrable factor
$<y>^{-2n}<z>^{-2n}<\eta>^{-2n}<\zeta>^{-2n}$ with $(d/2)<n\in\mathbb{N}$ and
get rid of the growing factors by the usual integration by parts trick using the
identities
\beq
\partial_{y_j}e^{-2i\sigma(Y,Z)}=2i\zeta_je^{-2i\sigma(Y,Z)},\ 
\partial_{z_j}e^{-2i\sigma(Y,Z)}=-2i\eta_je^{-2i\sigma(Y,Z)},
\eeq
\beq
\partial_{\eta_j}e^{-2i\sigma(Y,Z)}=-2iz_je^{-2i\sigma(Y,Z)},\ 
\partial_{\zeta_j}e^{-2i\sigma(Y,Z)}=2iy_je^{-2i\sigma(Y,Z)}.
\eeq
\end{proof}

\section*{Acknowledgements}
N.A. thanks the ``Simion Stoilow'' Institute of Mathematics of the Romanian
Academy for its hospitality during the final elaboration of this work and the
Laboratory LR11ES53 ``Alg\`{e}bre, G\'{e}ometrie et Th\'{e}orie Spectrale'' of
the University of Sfax and particularly Mondher Damak for their support. RP
thanks the University of Gafsa for its kind hospitality and acknowledges the
partial support from the Research Grant of the Romanian National
Authority for Scientific Research, CNCS-UEFISCDI, project number
PN-II-ID-PCE-2011-3-0131.

\end{document}